\documentclass{llncs}

\usepackage[pdftex]{graphicx}
\usepackage{url}
\usepackage{ascmac,bm}
\usepackage{amsmath,amssymb}
\usepackage{amsmath}

\usepackage{float}
\usepackage{algorithm,algorithmic}
\usepackage{extarrows}

\DeclareMathOperator{\tr}{Tr}
\DeclareMathOperator{\mtr}{\nat Tr}
\DeclareMathOperator{\ltr}{\nat Tr^N}

\newcommand{\defeq}{\xlongequal{\!\!\!\mathtt{def}\!\!\!}}
\newcommand{\defif}{\stackrel{\mathtt{def}}{\Longleftrightarrow}}
\newcommand{\ie}{{\it i.e.}}
\newcommand{\eg}{{\it e.g.}}
\newcommand{\cf}{{\it cf.}\,}
\newcommand{\etc}{{\it etc.}}

\newcommand{\dgp}[1]{{\mathcal G}^{#1}_A}
\newcommand{\dg}{\dgp{N}}
\newcommand{\CG}{{\mathcal G}}
\newcommand{\CK}{{\mathcal K}}

\newcommand{\nat}{\mathbb{N}}
\renewcommand{\int}{\mathbb{Z}}

\newcommand{\suff}{\mathrm{suff}}
\newcommand{\pref}{\mathrm{pref}}
\newcommand{\Walk}{\mathtt{walk}_{\dg}}
\newcommand{\Word}{\mathtt{word}_{\dg}}
\newcommand{\occv}[2]{|#1|_{(#2)}}
\newcommand{\sufv}[2]{|#1|_{(#2)}^{\mathtt{suff}}}
\newcommand{\diff}{\mathtt{diff}}
\newcommand{\dec}[1]{\Phi_{#1}}
\newcommand{\comp}[1]{\Psi_{#1}}
\newcommand{\walk}{\omega}
\renewcommand{\path}{\pi}
\newcommand{\cycle}{\gamma}
\newcommand{\walks}{{\mathcal W}}

\newcommand{\paths}{{\mathcal P}}
\newcommand{\cycles}{{\mathcal C}}
\newcommand{\emptyseq}{\emptyset}
\newcommand{\conn}{\odot}
\newcommand{\card}[1]{\#\!\left(#1\right)}

\newcommand{\wmix}[1]{L(#1)}
\newcommand{\from}[1]{\mathtt{from}(#1)}
\renewcommand{\to}[1]{\mathtt{into}(#1)}
\newcommand{\pfun}{\rightharpoonup}
\newcommand{\pleq}{\leq_{\mathtt{pt}}}

\begin{document}
\title{Note on the Infiniteness and Equivalence Problems for Word-MIX
Languages}
\author{Ryoma Sin'ya}
\institute{Akita University\\
\email{ryoma@math.akita-u.ac.jp}}
\maketitle

\begin{abstract}
In this note we provide a (decidable) graph-structural characterisation
of the infiniteness of $\wmix{w_1, \ldots, w_k}$,
where $\wmix{w_1, \ldots, w_k} = \{ w \in A^* \mid |w|_{w_1} = \cdots = |w|_{w_k}\}$ is the set of all words
that contain the same number of subword occurrences of parameter words
 $w_1, \ldots, w_k$.
We also provide the decidable characterisation of the equivalence for
 those languages.
Although those two decidability results are also obtained from more general known decidability results on unambiguous constrained automata,
this note tries to give a self-contained (without the
 knowledge about constrained automata) proof of the decidability.
\end{abstract}


\section{Introduction}\label{sec:intro}
Counting occurrences of letters in words is a major topic in formal
language theory and much ink has been spent on this topic.
Measuring the counting ability of a language class is
in this topic.
For example, Joshi et al.~\cite{joshi} suggested that the language MIX
$\{w \in \{a,b,c\}^* \mid |w|_a = |w|_b = |w|_c \}$ should not be in the class of
so-called mildly context-sensitive languages since it allows too much freedom in word order, so that
relations between MIX and several language classes have been
investigated (\eg,
indexed languages~\cite{marsh},
range concatenation languages~\cite{rcg},
tree-adjoining languages~\cite{kanazawa}, multiple
context-free languages~\cite{salvati}\cite{sorokin}, \etc). The Parikh
map is another rich example on this topic~\cite{parikh}.

In the recent work~\cite{Finn} by Colbourn et. al.,
the counting feature of MIX is generalised from the counting \emph{letter}
occurrences to the counting of \emph{word} occurrences.
They considered several problems for languages of the form
$\wmix{w_1,\ldots, w_k} = \{ w \in A^* \mid |w|_{w_1} = \cdots =
|w|_{w_k}\}$ which we call \emph{Word-MIX languages} (WMIX for short) in
this note.
It is interesting that the situation is drastically changed by this
generalisation.
The decidability of the infiniteness/equivalence turn to be non-trivial:
$\wmix{0,1,00,11}$ and $\wmix{0,1,01,10}$ are finite but
$\wmix{00,11,000,111}$ is infinite over $A = \{0,1\}$ (example from~\cite{Finn}), and
$\wmix{ab,ba,a}$ is infinite but
$\wmix{ab,ba,a,b}$
is finite over $A= \{a, b\}$ (these two examples appear again in Section~\ref{sec:ex}), for example.
In addition, while $\wmix{w_1, w_2}$ is always deterministic context-free (DCFL), it can also
be regular  ($\wmix{ab, ba} \subseteq \{a,b\}^*$
is regular, for example)~\cite{Finn}.
This kind of generalisation (from letter occurrences to word occurrences) is also considered in the context of the Parikh map~\cite{parikhmat}.
Colbourn et. al.~\cite{Finn} provided a necessary and sufficient condition
for $w_1$ and $w_2$ for these languages to be regular, and gave a
polynomial time algorithm for testing that condition.
The finiteness of $\wmix{w_1, w_2}$ is also considered in~\cite{Finn}
and they proved that, for any non-empty words $w_1, w_2 \in A^*$,
$\wmix{w_1, w_2}$ is finite if and only if the alphabet $A$ consists
of a single letter ($\ie, A = \{a\}$) and $w_1 \neq w_2$.
For more general case, allowing more than two parameter words $\wmix{w_1, \ldots, w_k} \, (k \geq 2)$,
 they give a sufficient condition for the infiniteness of $\wmix{w_1,
 \ldots, w_k}$ (Theorem~8 in \cite{Finn}):
 if all of $w_1, \ldots, w_k$ have the same length, then
 $\wmix{w_1, \ldots, w_k}$ is infinite.

For the fully general case, the decidability of both regularity and
infiniteness for WMIX languages can be derived from some known
results on \emph{constrained automata} (CA for short),
since \(\wmix{w_1, \ldots, w_k}\) is always recognised by a deterministic
CA, and its regularity and Parikh image are effectively
computable~\cite{UnCA}.
In this note, we provide a self-contained
(without the knowledge about constrained automata) description of a
 decidable, necessary and sufficient condition for the infiniteness of
 $\wmix{5w_1, \ldots, w_k}$, and give some open problems about the
 infiniteness.
Our proof is based on a \emph{combinatorics on walks in the de Bruijn
 graph}.
 The ($n$-dimensional) de Bruijn graph~\cite{deBruijn} can track all
 information of subword occurrences (of length at most $n$), hence it is
 a very useful tool for counting subword occurrences and related
 problems.
 The de Bruijn graph also played a key role in the proof of
 Theorem~8 in \cite{Finn}.

 The rest of this note consists as follows.
 In Section~\ref{sec:pre}, we give some preliminary definitions and
 propositions about words, orders, graphs and walks.
 Section~\ref{sec:decomp} investigates a simple decomposition method
 which decomposes a walk into a path and a sequence of cycles.
 This decomposition is useful for the proof of our main theorem.
 The main result of this note (Theorem~\ref{thm}), which states
 a decidable characterisation of the infiniteness of
 a WMIX language, is stated and proved in
 Section~\ref{sec:proof}, the decidability is explained with two
 examples in Section~\ref{sec:ex}.
 The decidability of the equivalence for two WMIX languages is also
 explained in Section~\ref{sec:equiv}.
 We end this note with list of open problems in Section~\ref{sec:future}.

\section{Preliminaries}\label{sec:pre}
For a set $X$, we denote by $\card{X}$ the cardinality of $X$.
We write $\card{X} = \infty$ if $X$ is an infinite set, and write
$\card{X} < \infty$ otherwise.
We denote by $\nat$ the set of natural numbers including $0$.
We call a mapping $M: X \rightarrow \nat$ multiset over $X$.

\subsection{Words and Orders}
For an alphabet $A$, we denote the set of all (resp. non-empty) words
over $A$ by $A^*$ (resp. $A^+$).
We write $A^n$ (resp. $A^{<n}$) the set of all words of length $n$
(resp. less than $n$).
For a pair of words $v, w \in A^*$, $|w|_v$ denotes the number of
subword occurrences of $v$ in $w$
\[
|w|_v \defeq \card{\{ (w_1, w_2) \in A^* \times
 A^* \mid w_1 v w_2 = w \}}.
\]
 For words $w_1, \ldots, w_k \in A^*$, we define
\begin{align*}
\wmix{w_1, \ldots, w_k} \defeq \{ w \in A^* \mid |w|_{w_1} = \cdots =
 |w|_{w_k} \}
\end{align*}
and call it the \emph{Word-MIX language of $k$ parameter words $w_1,
\ldots, w_k$} ((k-)WMIX for short).
For a word $w \in A^*$, we denote
the set of prefixes and suffixes of $w$ by
\begin{align*}
\pref(w) & \defeq \{ u \in A^* \mid
 uv = w \text{ for some } v \in A^*\} \\
\suff(w) & \defeq \{ v \in A^*
\mid uv = w \text{ for some } u \in A^* \} 
\end{align*}
and denote the length-$n$ ($n \leq |w|$) prefix and suffix of $w$
by $\pref_n(w)$ and $\suff_n(w)$, respectively.

A quasi order $\leq$ on a set $X$ is called \emph{well-quasi-order}
(\emph{wqo} for short) if any infinite sequence $(x_i)_{i \in \nat} \,
(x_i \in X)$ contains an increasing pair $x_i \leq x_j$ with $i < j$.
Let $\leq_1$ be a quasi order on a set $X_1$ and $\leq_2$ be a
 quasi order on a set $X_2$.
 The \emph{product order} $\leq_{1,2}$
 is a quasi order on $X_1 \times X_2$ defined by
\[
  (x_1, y_1) \leq_{1,2} (x_2, y_2) \defif x_1 \leq_1
 x_2 \text{ and } y_1 \leq_2 y_2.
\]
\begin{lemma}[\cf Proposition~6.1.1 in \cite{deLuca}]\label{wqo}
 Let $\leq_1$ be a wqo on a set $X_1$ and $\leq_2$ be a
 wqo on a set $X_2$. The product order $\leq_{1,2}$
 is again a wqo on $X_1 \times X_2$.
\end{lemma}
We list some examples of wqos below:
\begin{enumerate}
 \item The identity relation $=$ on any finite set $X$ is a wqo (\emph{the pigeonhole principle}).
 \item The usual order $\leq$ on $\nat$ is a wqo.
 \item The product order $\leq_m$ on $\nat^m$ is a wqo  for any
	   $m \geq 1$ (\emph{Dickson's lemma}), which is a direct corollary
	   of Lemma~\ref{wqo}.
 \item The point-wise order $\pleq$ on the multisets
	   $\nat^X$ ($M \pleq M' \defif M(x) \leq M'(x)$ for all $x \in X$)
	   over a finite set $X$ is a wqo
	   (just a paraphrase of Dickson's lemma).
 \end{enumerate}

\subsection{Graphs and Walks}
Let $\CG = (V, E)$ be a (directed) graph.
 We call a sequence of vertices $\walk = (v_1, \ldots, v_n) \in V^n \,
 (n \geq 1)$ \emph{walk} (from $v_1$ into $v_n$ in $\CG$) if $(v_i, v_{i+1}) \in E$ for each $i
 \in \{1, \ldots,  n-1\}$, and define the length of $\walk$ as $n-1$ and
 denote it by $|\walk|$.
 We denote by $\from{\walk}$ and $\to{\walk}$ the source $\from{\walk}
 \defeq v_1$ and the target $\to{\walk} \defeq v_n$ of $\walk$.
 $\walk$ is called an \emph{empty walk} if $|\walk| = 0$.
 If two walks $\walk_1 = (v_1, \ldots, v_m), \walk_2 = (v'_1, \ldots,
 v'_n)$ is connectable (\ie, $\to{\walk_1} = \from{\walk_2}$), we write
 $\walk_1 \conn \walk_2$ for the connecting walk $\walk_1 \conn \walk_2
 \defeq (v_1, \ldots, v_m, v'_2, \ldots, v'_n)$.
  A non-empty walk $\walk$ is called $\emph{loop}$ (on $\from{\walk}$) if $\from{\walk} = \to{\walk}$.
  A walk $(v_1, \ldots, v_n)$ is called \emph{path}
 if $v_i \neq v_j$ for every $i, j \in \{1, \ldots, n\}$ with $i \neq
 j$.
 A loop $(v, v_1, \ldots, v_n, v)$ is called \emph{cycle} if
 $(v, v_1, \ldots, v_n)$ is a path.
 We use the metavariable $\path$ for a path, and the metavariable $\cycle$
 for a cycle.
 For a cycle $\cycle$ and $n \geq 1$, we write $\cycle^n$ for
 the loop which is an $n$-times repetition of $\cycle$.
We denote by $\walks(\CG), \paths(\CG),$ and by
 $\cycles(\CG)$ the set of all walks, paths and cycles in $\CG$.
 Note that $\walks(\CG)$ is infinite in general, but
 $\paths(\CG)$ and $\cycles(\CG)$ are both finite if $\CG$ is finite
 (\ie, $\card{V} < \infty$).

The \emph{$N$-dimensional de Bruijn graph} $\dg = (A^N, E)$ over $A$ is a graph whose vertex set $A^N$ is the set of words
of length $N$ and the edge set $E$ is defined by
\[
 E \defeq \{ (a v, v b) \mid a,b \in A, v \in A^{N-1} \}.
\]
The case $N = 2$ is depicted in Fig.~\ref{fig:dg}.

\begin{figure}[t!]
 \centering\includegraphics[width=0.9\columnwidth]{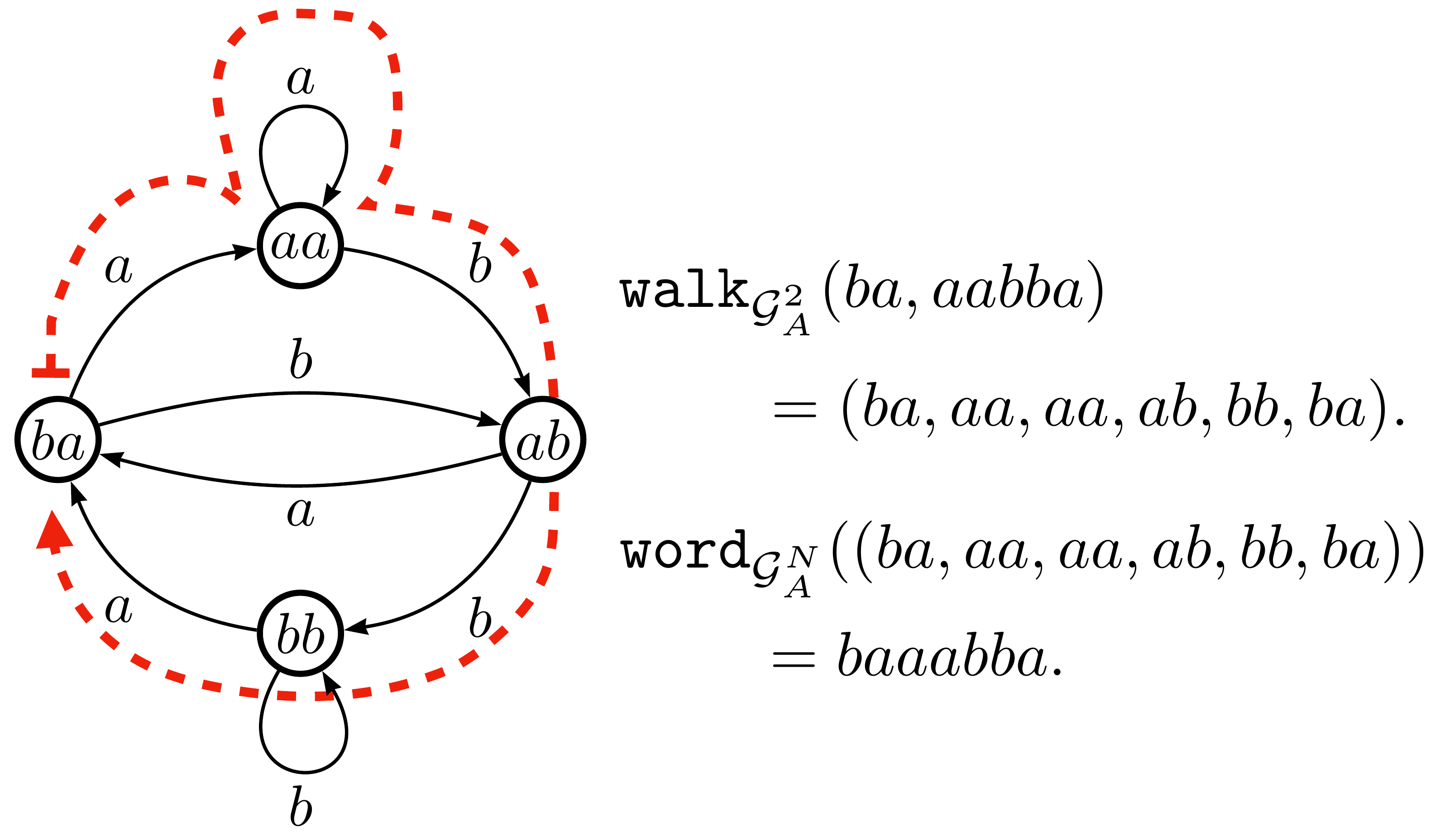}
 \caption{The 2-dimensional de Bruijn graph $\dgp{2}$ over $A = \{a,
 b\}$,
 a walk $(ba, aa, aa, ab, bb, ba)$ (dotted red arrow) on $\dgp{2}$
 and its corresponding word $baaabba$.}
 \label{fig:dg}
\end{figure}

Let $v$ be a vertex of $\dg$.
A word $w = a_1 \cdots a_m \in A^+$ induces the walk $(v, v_1, \ldots,
v_m)$ (where $v_i = \suff_n(v \, \pref_i(w))$) in $\dg$, and we denote
it by $\Walk(v, w)$.
Conversely, a walk $\walk = (v_1, \ldots, v_n)$ in $\dg$ induces the word
$v_1 \suff_1(v_2) \cdots \suff_1(v_n) \in A^*$, and we denote it by
$\Word(\walk)$ (see Fig.~\ref{fig:dg}).
For words $w, w_1, \ldots, w_k \in A^*$
and a walk $\walk = (v_0, v_1, \ldots, v_n) \in \walks(\dg)$, we define
the following vectors in $\nat^k$:
\begin{align*}
\sufv{w}{w_1, \ldots, w_k} \defeq \, &
 (c_1, \ldots, c_k)
 \text{ where } c_i = 1 \text{ if }
 w_i \in \suff(w), c_i = 0 \text{ otherwise},\\
\occv{w}{w_1, \ldots, w_k} \defeq \, &
 (|w|_{w_1}, \ldots, |w|_{w_k})
 \qquad
 \occv{\walk}{w_1, \ldots, w_k} \defeq 
 \sum_{i = 1}^{n} \sufv{v_i}{w_1, \ldots, w_k}.
 \end{align*}
We notice that the range of the summation in the above definition
of $\occv{\walk}{w_1, \ldots, w_k}$ \emph{does not contain $0$},
hence $\occv{\walk}{w_1, \ldots, w_k} = (0, \ldots, 0)$ if
$\walk$ is an empty walk $\walk = (v_0)$.
The next proposition states a basic property of $\dg$.

\begin{proposition}\label{prop:dg}
Let $w_1, \ldots, w_k \in A^*$ and $N = \max(|w_1|, \ldots, |w_k|)$.
For any pair of words $v,w \in A^*$ such that $|v| = N$, we have
\[
\occv{vw}{w_1, \ldots, w_k}
=
\occv{v}{w_1, \ldots, w_k}
 + \occv{\walk}{w_1, \ldots, w_k}
\]
 where $\walk = \Walk(v, w)$.
\end{proposition}
\begin{proof}
Straightforward induction on the length of $w$.
\end{proof}

\section{Path-Cycle Decomposition of Walks}\label{sec:decomp}
In this section, we provide a simple method which
decomposes, in left-to-right manner, a walk $\walk$ into a (possibly
empty) path $\path$ and a sequence of cycles $\Gamma$ (Fig.~\ref{fig:decomp}).
This decomposition is probably folklore but useful for our main proof in
the next section.
We also introduce in this section the notion of multi-traces and traces
of walks,
which play crucial role in the characterisation of the infiniteness and
equivalence for WMIX languages.

Let $\CG = (V, E)$ be a graph.
For a pair of sequences of cycles $\Gamma_1 = (\cycle_1, \ldots,
\cycle_n),
\Gamma_2 = (\cycle'_1, \ldots, \cycle'_m)$, we write $\Gamma_1.\Gamma_2$
for the concatenation $(\cycle_1, \ldots, \cycle_n, \cycle'_1, \ldots,
\cycle'_m)$. When $\Gamma_1 = (\cycle)$ we simply write
$\cycle.\Gamma_2$ for $\Gamma_1.\Gamma_2$
We write $\emptyseq$ for the empty sequence of cycles.
For $\Gamma = (\cycle_1, \ldots, \cycle_n)$, we denote
by $\Gamma(i)$ for the $i$-th component $\cycle_i$ of $\Gamma$, and 
denote by $|\Gamma|_\cycle$ the number $\card{\{i \mid \Gamma(i) =
\cycle\}}$ of occurrences of $\cycle$ in $\Gamma$.
For a walk $\walk = (v_1, \ldots, v_n)$, we denote by $V(\walk)$ the set of all vertices
appeared in $\walk$: $V(\walk) \defeq \{ v_1, \ldots, v_n \}$.

\begin{figure}[t]
 \centering\includegraphics[width=0.86\columnwidth]{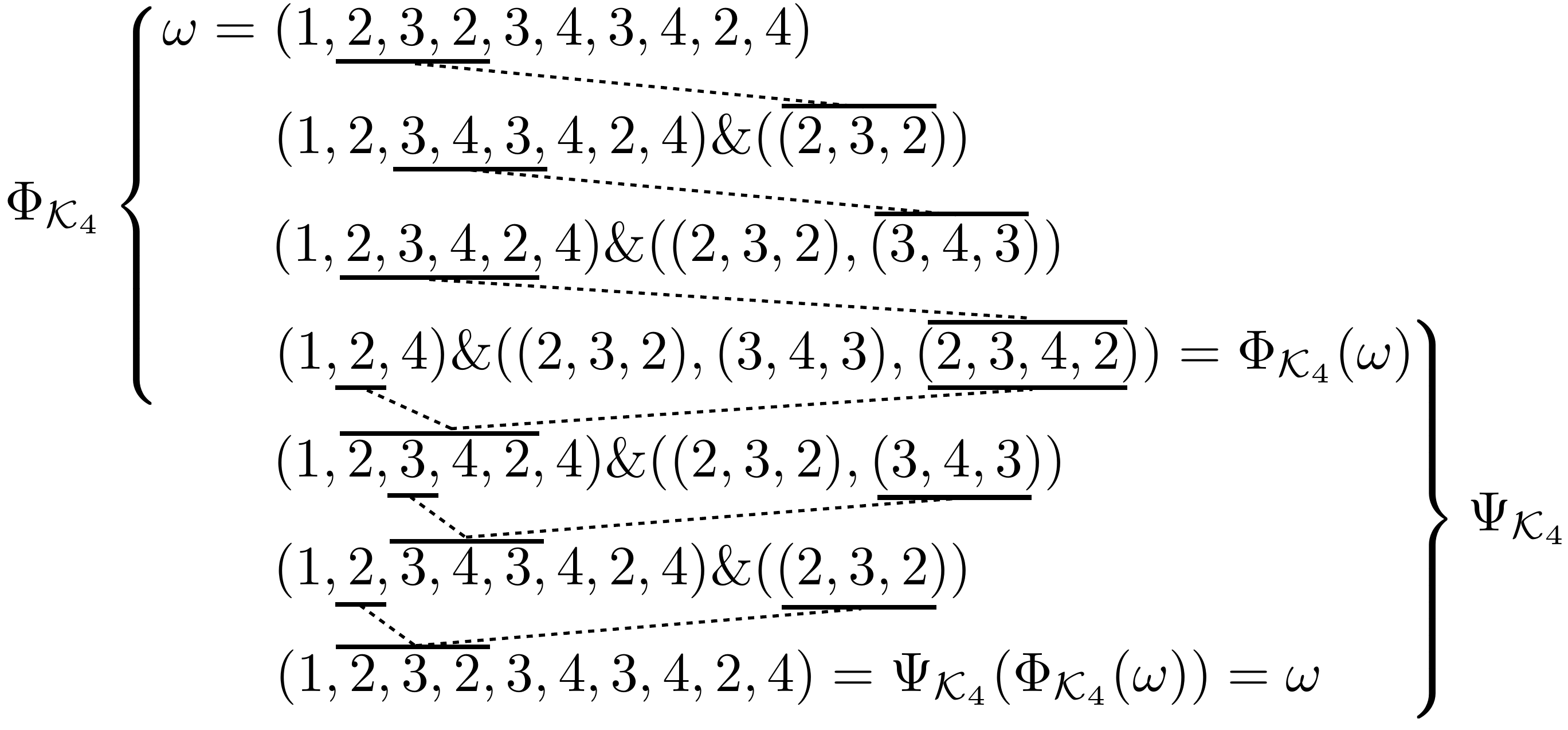}
 \caption{Computation of $\dec{\CK_4}$ and $\comp{\CK_4}$}
 \label{fig:decomp}
\end{figure}

We then define a decomposition function $\dec{\CG}: \walks(\CG)
\rightarrow \paths(\CG) \times \cycles(\CG)^*$ inductively as follows:
\begin{align*}
 \dec{\CG}( (v) ) & \defeq ((v), \emptyseq)\\
 \dec{\CG}(\walk \conn (v, v')) & \defeq
 \begin{cases}
  (\path \conn (v, v'), \,\, \Gamma) \qquad\text{ if } v' \notin V(\path),\\
  ((v_1, \ldots, v_{m-1}, v'), \,\, \Gamma . (v', v_{m+1}, \ldots, v, v'))\\
  \phantom{(\path \conn (v, v'), \,\, \Gamma)}
  \qquad\text{ if }
  \path = (v_1, \ldots, v_{m-1}, v', v_{m+1}, \ldots, v)
 \end{cases}\\
& \quad\qquad \text{ where } (\path, \Gamma) = \dec{\CG}( \walk ).
\end{align*}
Conversely, we define a composition (partial) function $\comp{\CG}:
{\walks(\CG) \times \cycles(\CG)^* \pfun \walks(\CG)}$
inductively as follows:
\begin{align*}
 \comp{\CG}(\walk, \emptyseq) & \defeq \walk\\
 \comp{\CG}(\walk, \cycle.\Gamma) &\defeq 
  \begin{cases}
   \path \conn \cycle \conn \walk' & \text{if } 
   \path \conn (v) \conn \walk' = \comp{\CG}(\walk, \Gamma)\\
   & \phantom{\text{if }} \text{ where } \from{\cycle} = v \text{ and } \path \text{ is a path to } v, \\
  \text{undefined} & \text{if } \from{\cycle} \notin V(\comp{\CG}(\walk, \Gamma)).
  \end{cases}
\end{align*}

We list some important properties of $\dec{\CG}$ and $\comp{\CG}$.
\begin{proposition}\label{prop:dec}
 Let $\CG = (V, E)$ be a graph and $\walk \in \walks(\CG)$.
 Then the followings hold for $(\path, \Gamma) = \dec{\CG}(\walk)$:
 \begin{enumerate}
  \item $\path$ is a path in $\CG$. \label{decpi}
  \item $\Gamma$ is a sequence of cycles in $\CG$. \label{deccy}
  \item $|\walk| = |\path| + \sum_{i = 1}^{|\Gamma|} |\Gamma(i)|$ \label{declen}
  \item $\walk = \comp{\CG}(\path, \Gamma)$, \ie,
		$\comp{\CG} \circ \dec{\CG}$ is the identity
		function on $\walks(\CG)$. \label{decinv}
 \end{enumerate}
\end{proposition}
\begin{proof}
\eqref{decpi}--\eqref{deccy} are obvious by the definition.
\eqref{declen}--\eqref{decinv} can be shown by an easy induction on the length of $\walk$.
\end{proof}

\begin{proposition}\label{prop:occinv}
Let $w_1, \ldots, w_k \in A^*$ and $N = \max(|w_1|, \ldots, |w_k|)$.
 For any $\walk$ in $\walks(\dg)$,
 \[
 \occv{\walk}{w_1, \ldots, w_k} =
 \occv{\path}{w_1, \ldots, w_k} + \sum_{i = 1}^{|\Gamma|}
 \occv{\Gamma(i)}{w_1, \ldots, w_k}
 \]
 holds where $(\path, \Gamma) = \comp{\dg}(\path, \Gamma)$.
\end{proposition}
\begin{proof}
Straightforward induction on the length of $\walk$.
\end{proof}

\begin{example}\label{ex:decomp}
Consider the complete graph $\CK_4 = (V_4 = \{1,2,3,4\}, E_4 = V_4
 \times V_4)$ of order $4$ and a walk $\walk = (1,2,3,2,3,4,3,4,2,4)$.
 The result of decomposition is $\dec{\CK_4}(\walk) = (\walk = (1,2,4),
 \Gamma = ((2,3,2), (3,4,3), (2,3,4,2)))$.
 All intermediate computation step of $\dec{{\cal K}_4}(\walk)$
 and $\comp{\CK_4}(\dec{\CK_4}(\walk))$ are drawn in
 Fig.~\ref{fig:decomp} (in the figure we denote by $\path\&\Gamma$ a
 pair $(\path, \Gamma)$ for visibility).
 It is clear that the all conditions in Proposition~\ref{prop:dec} are
 satisfied ($|\walk| = 9 = 2+2+2+3 = |\walk| + \sum_{i = 1}^{3}|\Gamma(i)|$, in particular).
\end{example}

\subsection{Multi-Traces and Traces}
For a walk $\walk$ in a graph $\CG$,
we define the \emph{multi-trace} $\mtr(\walk): \paths(\CG) \cup \cycles(\CG)
\rightarrow \nat$ of a walk $\walk$ as the following multiset over paths and cycles:
\begin{align*}
 (\mtr(\walk))(\path) \defeq \begin{cases}
							 1 & \text{ if } \path = \path_\walk\\
							 0 & \text{ otherwise}
							\end{cases}
							&\qquad\qquad\qquad
							(\mtr(\walk))(\cycle) \defeq |\Gamma|_{\cycle}\\
\text{where } & (\path_\walk, \Gamma) = \dec{\CG}(\walk).
\end{align*}
We define the \emph{trace} $\tr(\walk)$ of a walk $\walk$ in $\CG$ as
the following set of paths and cycles:
\begin{align*}
 \tr(\walk) \defeq \{ \path \in \paths(\CG) \mid
 (\mtr(\walk))(\path) \neq 0 \} \cup
 \{ \cycle \in \cycles(\CG) \mid (\mtr(\walk))(\cycle) \neq 0 \}.
\end{align*}
Intuitively, the multi-trace of $\walk$ in $\CG$ is obtained by
forgetting the ordering of the decomposition result $(\walk, \Gamma) =
\dec{\CG}(\walk)$ of $\walk$,
and the trace of $\walk$ is obtained by forgetting the multiplicity from
the original multi-trace (see Fig.~\ref{fig:rel} for the relation).

\begin{figure}[t]
 \centering\includegraphics[width=1\columnwidth]{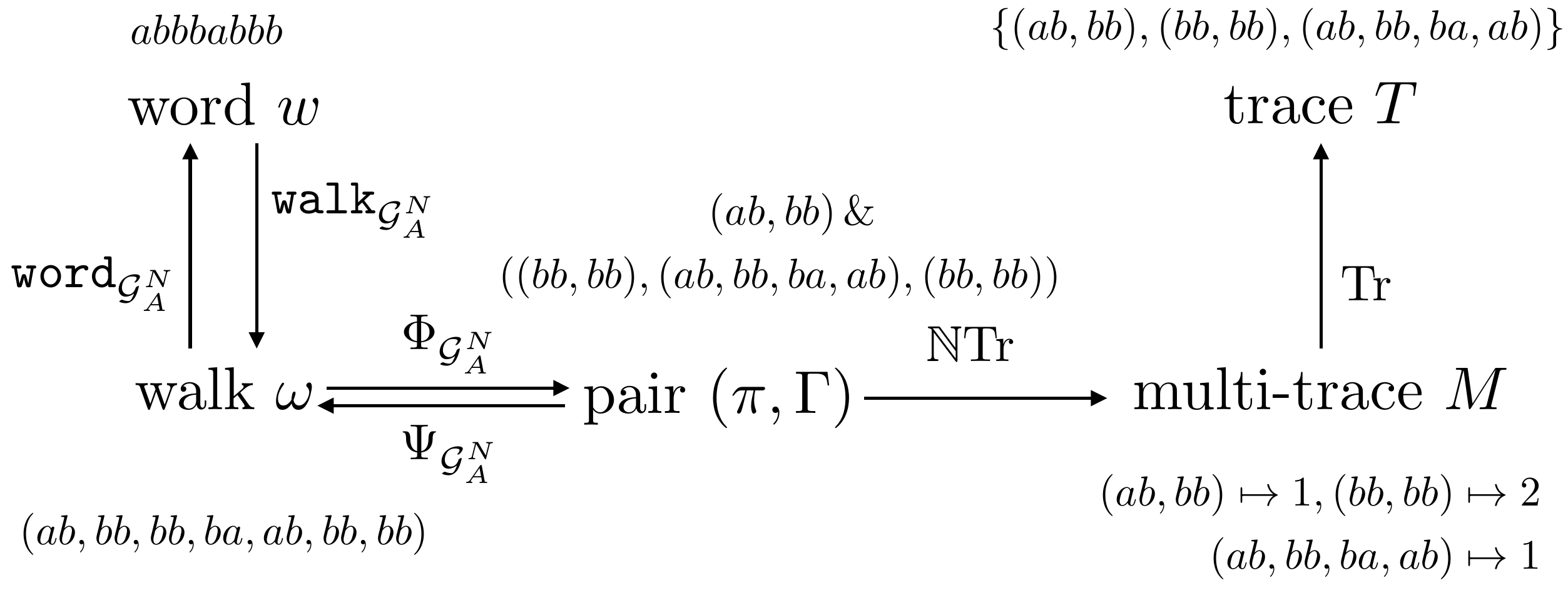}
 \caption{Relations between words, walks and (multi-)traces ($N = 2$ for the examples).}
 \label{fig:rel}
\end{figure}

Since Item~\ref{declen} in Proposition~\ref{prop:dec} and
Proposition~\ref{prop:occinv}
do not depend on the order of a sequence $\Gamma$, one can easily observe that the
following proposition holds by the definition of $\mtr(\walk)$.
\begin{proposition}\label{prop:mtrace}
Let $w_1, \ldots, w_k \in A^*$ and $N = \max(|w_1|, \ldots, |w_k|)$.
For any $\walk$ in $\walks(\dg)$, we have
\[
\occv{\walk}{w_1, \ldots, w_k} = \!\!\! \sum_{\path \in
 \paths(\dg)} \!\!\!
 (\mtr(\walk))(\path) \cdot \occv{\path}{w_1, \ldots, w_k} + \!\!\!
 \sum_{\cycle \in \cycles(\dg)} \!\!\!
 (\mtr(\walk))(\cycle) \cdot \occv{\cycle}{w_1, \ldots, w_k}
\]
 and
 \[
|\walk| = \sum_{\path \in
 \paths(\dg)} \!\!\!
 (\mtr(\walk))(\path) \cdot |\path| + \!\!\!
 \sum_{\cycle \in \cycles(\dg)} \!\!\!
 (\mtr(\walk))(\cycle) \cdot |\cycle|.
 \]
\end{proposition}
For a set $T \subseteq \paths(\CG) \cup \cycles(\CG)$,
the following lemma states that we can effectively test whether $T$ is a
trace or not (see Fig.~\ref{fig:trace} for the intuition).
\begin{lemma}\label{lem:trace}
Let $T \subseteq \paths(\CG) \cup \cycles(\CG)$ be a set of paths and
 cycles in a graph $\CG = (V, E)$. The followings are equivalent:
 \begin{enumerate}
  \item $T$ is a trace of some walk in $\CG$.
  \item $T$ can be written as $T = \{\path\} \cup \{\cycle_1, \ldots, \cycle_m\}$
		such that
		(i) $V(\cycle_m) \cap V(\path) = \{\from{\cycle_m}\}$ and (ii)
		for every $i \in \{1, \ldots, m-1\}$,
		$V(\path') \cap V(\cycle_i) = \{\from{\cycle_i}\}$
		where $\path' \conn (\from{\cycle_i})
		\conn \walk = \comp{\CG}(\path, (\cycle_{i+1}, \ldots, \cycle_m)).$ \label{con:trace}
 \end{enumerate}
\end{lemma}

\begin{figure}[t]
 \centering\includegraphics[width=0.85\columnwidth]{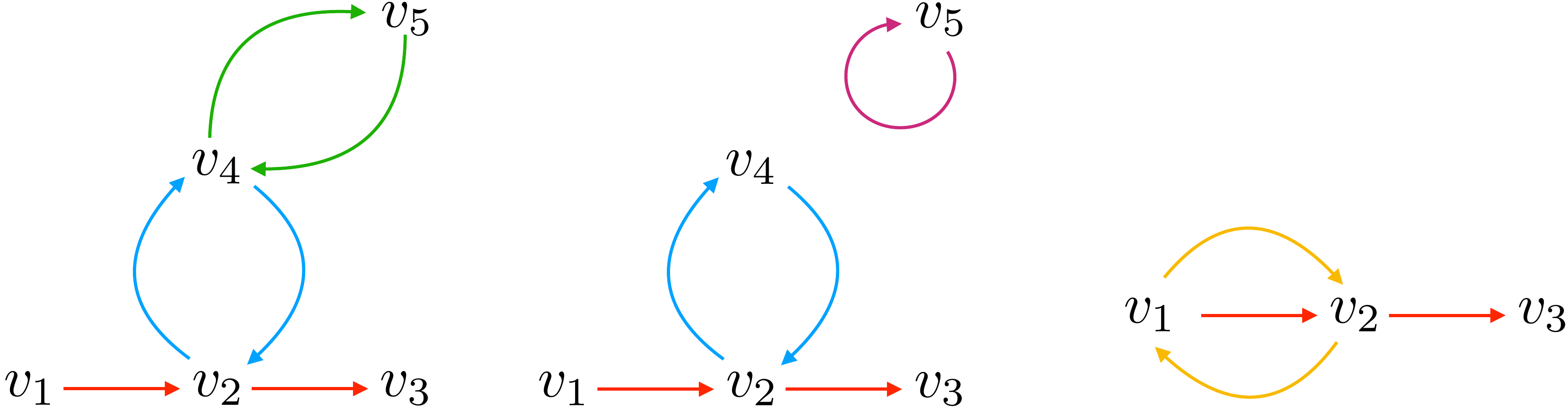}
 \caption{({\bf left}) $\{\path \!\!=\!\! (v_1,v_2,v_3)\} \cup \{ (v_4, v_5, v_4),
 \cycle_2 \!\!=\!\! (v_2, v_4, v_2) \}$ is a trace.
 ({\bf middle}) $\{
 \path \} \cup \{ (v_5, v_5), \cycle_2\}$ is not connected thus not a trace. ({\bf right})
 $\{\path\} \cup \{(v_1, v_2, v_1)\}$ is a trace and
 $\{\path\} \cup \{(v_2, v_1, v_2)\}$ is connected but not a trace (from Condition~(2-i) in Lemma~\ref{lem:trace}).
 }
 \label{fig:trace}
\end{figure}
\section{Characterisation of the Finiteness}\label{sec:proof}
For a vector $\bm{v} = (c_1, \ldots, c_k) \in \nat^k$, we define
\[
 \diff(\bm{v}) \defeq \sum_{i = 1}^k (\max(c_1, \ldots, c_k) - c_i).
\]
Observe that $w \in \wmix{w_1, \ldots, w_k}$ if and only if
$\diff(\occv{w}{w_1, \ldots, w_k}) = 0$.

We now ready to state and prove the main result.
\begin{theorem}\label{thm}
 Let $w_1, \ldots, w_k \in A^*$ and $N = \max(|w_1|, \ldots, |w_k|)$.
 Then the followings are equivalent:
 \begin{enumerate}
  \item $\wmix{w_1, \ldots, w_k}$ is infinite.
  \item There exists a trace $T \subseteq \paths(\dg) \cup \cycles(\dg)$ that satisfies the following two
		conditions.
		By Lemma~\ref{lem:trace}, we can assume without loss of generality that $T$ is of the form $T =
		\{\path\} \cup \{\cycle_1, \ldots, \cycle_m\}$ that satisfies
		Condition~\eqref{con:trace} in Lemma~\ref{lem:trace}.
		\\
		\noindent{\bf (balance condition)}
		there exist positive coefficients
					$x_1, \ldots, x_m \in \nat$, $x_i \geq 1$ for
					each $i \in \{1, \ldots, m\}$, such that
					\[
		\diff\left( \occv{\from{\path}}{w_1, \ldots, w_k} + \occv{\path}{w_1,
		\ldots, w_k} + \sum_{i = 1}^m x_i \cdot \occv{\cycle_i}{w_1,
		\ldots, w_k} \right) = 0.
					\]
		\noindent{\bf (pumping condition)}
					there exist coefficients
					$y_1, \ldots, y_m \in \nat$, not all zero, such that
					\[
					\diff\left( \sum_{i = 1}^{m} y_i \cdot \occv{\cycle_i}{w_1,	\ldots, w_k} \right) = 0.
					\]
 \end{enumerate}
\end{theorem}
\begin{proof}
 The direction $(2) \Rightarrow (1)$ is relatively easy.
 Intuitively, the balance condition ensures
 the existence of a word $v u_0 \in A^*$ such that $|v| = N$, $v u_0 \in
 \wmix{w_1, \ldots, w_k}$, and the pumping condition further ensures
 $v u_0$ is ``pumpable'' in some sense, which implies the infiniteness
 of $\wmix{w_1, \ldots, w_k}$.
 We prove this intuition. Assume a
 trace $T = \{\path \} \cup \{\cycle_1, \ldots,
 \cycle_m\}$ satisfies the balance and pumping conditions and let $v = \from{\path}$.
 Since $T$ is a trace, $\walk = \comp{\dg}(\path, (\cycle_1, \ldots,
 \cycle_m))$ is defined by Lemma~\ref{lem:trace}.
 Let $x_1, \ldots, x_m$ be positive coefficients that satisfy
 the balance condition, and $y_1, \ldots, y_m$ be coefficients, not
 all zero, that satisfy the pumping condition.
 For each $n \geq 0$, we define $\walk_n \defeq \comp{\dg}(\path,
 (\cycle_1^{x_1+n \cdot y_1}, \ldots, \cycle_m^{x_m+n \cdot y_m}))$
 and $v u_n \defeq \Word(\walk_n)$.
 Note that, since every cycle has at least length one and by
 Item~\ref{declen} in Proposition~\ref{prop:dec}, $|\walk_n| < |\walk_{n'}|$ and hence $|v u_n| < |v u_{n'}|$ holds for every $n < n'$.
 Combining Proposition~\ref{prop:dg}, Proposition~\ref{prop:occinv} and
 the balance condition, we obtain
\begin{align*}
 & \diff(\occv{v u_0}{w_1,\ldots,w_k})
 = \diff(\occv{v}{w_1,\ldots,w_k} +
 \occv{\walk_0}{w_1,\ldots,w_k}) \\
 = \, & \diff\left(
 \occv{v}{w_1,\ldots,w_k}
 + \occv{\path}{w_1,\ldots,w_k}
 + \sum_{i = 1}^{m} x_i \cdot \occv{\cycle_i}{w_1,\ldots,w_k} \right)  = 0.
\end{align*}
 Moreover, by Proposition~\ref{prop:dg}, Proposition~\ref{prop:occinv} and the pumping condition we have
\begin{align*}
 & \diff(\occv{v u_n}{w_1,\ldots,w_k}) =  \diff(\occv{v}{w_1,\ldots,w_k} + \occv{\walk_n}{w_1,\ldots,w_k})\\
= \, & 
 \diff\left(\occv{v}{w_1,\ldots,w_k} + 
 \occv{\walk_0}{w_1, \ldots, w_k} + n \sum_{i = 0}^{m}
 y_i \cdot \occv{\cycle_i}{w_1, \ldots, w_k}
 \right) = 0
\end{align*}
 for any $n \geq 1$.
 This means that every distinct word $v u_n$ is in $\wmix{w_1, \ldots, w_k}$, hence $\wmix{w_1, \ldots, w_k}$ is infinite.

We then prove the opposite direction $(1) \Rightarrow (2)$.
 Assume $\card{\wmix{w_1, \ldots, w_k}} = \infty$.
 Since $\wmix{w_1, \ldots, w_k}$ is infinite and hence it contains an arbitrary long word, we can
 take an infinite sequence $(v_i u_i)_{i \in \nat}$ of words from $\wmix{w_1, \ldots, w_k}$ that satisfies
 $|v_i| = N \, (i \in \nat)$ ,$v_i u_i \in
 \wmix{w_1, \ldots, w_k} \, (i \in \nat)$ and $|v_i u_i| < |v_j u_j| \, (i < j)$.
 Now consider an infinite sequence of multi-traces
 \[
(M_i)_{i \in \nat} \defeq (\mtr(\Walk(v_i, u_i)))_{i \in \nat}.
 \]
 Since the point-wise order on the multisets over any finite set is a wqo (thanks to Dickson's lemma) and
 $\paths(\dg) \cup \cycles(\dg)$ is finite,
 $(M_i)_{i \in \nat}$ contains an
 increasing pair $M_{i_1} \pleq M_{i_2}$
 with $i_1 < i_2$.
 Define $T_i \defeq \tr(\Walk(v_i, u_i) \, (i \in \{i_1, i_2\})$.
 We notice that, since
 $M_{i_1} \pleq M_{i_2}$
 and every multi-trace contains exactly one path as its
 element,
 $T_{i_1} \subseteq T_{i_2}$ and $(T_{i_1} \cap T_{i_2}) \cap
 \paths(\dg) = \{\path\}$ for some path $\path$.
 Since $|v_i u_i| < |v_j u_j|$ holds for every $i < j$ by definition, we can
 deduce $M_{i_1} \neq M_{i_2}$ by Proposition~\ref{prop:mtrace} and thus
 we have
\begin{align}
 M_{i_1}(\cycle) < M_{i_2}(\cycle) \text{ for some cycle } \cycle \in
 T_{i_2}. \tag{$\bigstar$}\label{star}
 \end{align}
 Let $m$ be the number of cycles in $T_{i_2}$ and
 write $T_{i_2} = \{\path\} \cup \{ \cycle_1, \ldots, \cycle_m \}$.
 Define $x_i \defeq M_{i_2}(\cycle_i)$ and
 $x'_i \defeq M_{i_1}(\cycle_i)$ for each $i \in \{1, \ldots, m\}$.
 Clearly $x_i \geq 1$
 and $x_i \geq x'_i$ hold for every $i \in \{1, \ldots, m\}$ by the definition.
 By Proposition~\ref{prop:dg} and Proposition~\ref{prop:mtrace}, we have
 \begin{align*}
  & \diff \left(\occv{\from{\path}}{w_1, \ldots, w_k} +
  \occv{\path}{w_1, \ldots, w_k} + \sum_{i = 1}^{m} x_i
  \cdot \occv{\cycle_i}{w_1, \ldots, w_k}
  \right) \\
=\,&  \diff\left(\occv{\from{\path}}{w_1, \ldots, w_k} + \occv{\path}{w_1, \ldots, w_k} + \sum_{i = 1}^{m} x'_i
 \cdot \occv{\cycle_i}{w_1, \ldots, w_k}\right)\\
  =\, & \diff(\occv{\from{\path} u_{i_2}}{w_1,\ldots, w_k}) = \diff(\occv{\from{\path}
  u_{i_1}}{w_1, \ldots, w_k}) =  0,
 \end{align*}
that is, the balance condition is satisfied.
In addition, from the above equation we obtain
\begin{align*}
 \diff\left( \sum_{i = 1}^{m} (x_i - x'_i) \cdot
 \occv{\cycle_i}{w_1, \ldots, w_k} \right) = 0
\end{align*}
 because for any $\bm{v}$ such that $\diff(\bm{v}) = 0$, $\diff(\bm{v} + \bm{v'}) = 0$ if
 and only if $\diff(\bm{v}') = 0$.
 By Condition~\eqref{star},
 coefficients $(x_1 - x'_1)$, \ldots, $(x_m - x'_m)$ are not all zero, the pumping condition is satisfied. This ends the proof. \qed
\end{proof}

\subsection{Decidability and Examples}\label{sec:ex}
The decision problem whether both balance and pumping condition are
satisfied for a given trace $T \subseteq \paths(\dg) \cup \cycles(\dg)$ in $\dg$
can be reduced into $\mathrm{\Sigma}_1$-formula (existential formula) of
Presburger arithmetic (see the examples in below). The set of traces in $\dg$ is clearly
finite and effectively enumerable (due to Lemma~\ref{lem:trace}), in addition. Thus we obtain the following
corollary.
\begin{corollary}\label{cor:inf}
For all words $w_1, \cdots, w_k \in A^*$, it is decidable whether
$\wmix{w_1, \ldots, w_k}$ is infinite or not.
\end{corollary}
\begin{proof}
Enumerate possible traces in $\dg$
 and check whether there is a trace that satisfies both balance and pumping condition.
\end{proof}

\begin{example}\label{ex1}
 Consider the language $\wmix{ab, ba, a}$ over $A = \{a, b\}$,
 $\max(|ab|, |ba|, |a|) = 2$ and the 2-dimensional de Bruijn graph
 $\dgp{2}$ shown in Fig.~\ref{fig:dg}.
 We claim that a trace $T_1 = \{\path_1 = (ba, ab)\} \cup \{ \cycle_1 = (ba, ab, ba)\}$ satisfies both balance and pumping condition.
 One can easily observe that
 \begin{align*}
  \occv{ba}{ab,ba,a} = (0,1,1) \quad
  \occv{\path_1}{ab,ba,a} = (1,0,0) \quad
  \occv{\cycle_1}{ab,ba,a} = (1,1,1)
 \end{align*}
 and hence the coefficient $x_1 = 1$ simultaneously satisfies the two
 condition stated in (2) of Theorem~\ref{thm}.
 For each $n \geq 1$, by Proposition~\ref{prop:dg} the word
 $ba (ba)^n b = \Word(\cycle_1^n \conn \path_1)$ is in $\card{\wmix{ab, ba, a}}
 = \infty$. Hence $ba (ba)^+ b \subseteq \wmix{ab, ba, a}$ and
 $\card{\wmix{ab, ba, a}} = \infty$.
\end{example}

\begin{example}\label{ex2}
Next consider another language $\wmix{ab, ba, a, b}$ over $A = \{a, b\}$,
 $\max(|ab|, |ba|, |a|, |b|) = 2$ and again the 2-dimensional de Bruijn graph
 $\dgp{2}$ shown in Fig.~\ref{fig:dg}.
 In contrast with Example~\ref{ex1}, the trace
 $T_1 = \{\path_1 = (ba, ab)\} \cup \{\cycle_1 = (ba, ab, ba)\}$
 does not satisfy the balance condition any more (even it still satisfies the
 pumping condition).
 We have
 \begin{align*}
  \occv{ba}{ab,ba,a,b} =& (0,1,1,1) \qquad
  \occv{\path_1}{ab,ba,a,b} = (1,0,0,1) \\
  \occv{\cycle_1}{ab,ba,a,b} =& (1,1,1,1)
 \end{align*}
We can formally prove that there is no positive coefficient $x_1 \in
 \nat \, (x_1 > 0)$ that satisfies the balance condition,
 since the existence of such coefficients can be expressed in the
 following $\mathrm{\Sigma}_1$-formula of Presburger arithmetic
  \begin{align*}
   \phi_{T_1} \defeq
   \exists c \Bigl( \exists x_1 \bigl(& x_1 > 0 \, \land \theta_{T_1}^{ab} = c \land \theta_{T_1}^{ba} = c
   \land \theta_{T_1}^{a} = c \land
   \theta_{T_1}^{b} = c \bigr) \Bigr)\\
   \equiv
 \exists c \Bigl( \exists x_1 \bigl(& x_1 > 0 \, \land \\
  & (0+1+x_1) = c \land (1+0+x_1) = c \, \land\\
  & (1+0+x_1) = c \land (1+1+x_1) = c \quad \bigr) \Bigr)
  \end{align*}
 where $\theta_{T_1}^{w}$ is a subexpression defined by
\[
 \theta_{T_1}^{w} \defeq \occv{ba}{w} + \occv{\path_1}{w} +
 \underbrace{x_1 + \cdots + x_1}_{\occv{\cycle_1}{w}
 \text{ times}}.
\]
 $\phi_{T_1}$ can be algorithmically verified to be not valid since
 the validity of a first-order formula of Presburger arithmetic is
 decidable (\cf Section~6.2 of \cite{sipser}).
 We can algorithmically verify, by using the same reduction into $\mathrm{\Sigma}_1$-formulae of Presburger arithmetic, that
 no trace in $\dgp{2}$ satisfies both balance
 and pumping condition.
 Thus $\card{\wmix{ab,ba,a,b}} < \infty$ by Theorem~\ref{thm}.
\end{example}
\section{Characterisation of the Equivalence}\label{sec:equiv}
In the previous section, multi-traces and traces play crucial role for
the characterisation of the finiteness.
Multi-traces are also important for the characterisation of the \emph{equivalence} of WMIX languages which is given here.
Before stating the main statement, we lift the notion of
traces of walks to one of languages.
For a language $L \subseteq A^*$, we define $\mtr(L)$
 the \emph{multi-trace of a language $L$ (of order $N$)} as
 \[
  \ltr(L) \defeq \{ \mtr(\walk) \mid \walk = \Walk(v, u), |v| = N, vu \in L\}.
 \]

The following theorem states that any WMIX language is
 completely determined by its multi-trace (excluding shorter part $A^{<N}$).
\begin{theorem}~\label{thm:equiv}
  Let $w_1, \ldots, w_k, w'_1, \ldots, w'_{k'} \in A^*$ and $N =
 \max(|w_1|, \ldots, |w_k|, |w'_1|, \ldots, |w'_{k'}|)$.
 Then
 $\wmix{w_1, \ldots, w_k} = \wmix{w'_1, \ldots, w'_{k'}}$ if and only if
 \[
  \wmix{w_1, \ldots, w_k} \cap A^{<N} = \wmix{w'_1, \ldots, w'_{k'}} \cap
 A^{<N}
 \]
 and
 \[
 \ltr(\wmix{w_1, \ldots, w_k}) = \ltr(\wmix{w'_1, \ldots, w'_{k'}}).
 \]
\end{theorem}
\begin{proof}
 The ``only-if''-part is trivial. We prove the ``if''-part by contraposition.
 Assume $\wmix{w_1, \ldots, w_k} \neq \wmix{w'_1, \ldots, w'_{k'}}$.
 Then we can assume that there is some word $w$ such that $w \in \wmix{w_1,
 \ldots, w_k}$ but $w \notin \wmix{w'_1, \ldots, w'_{k'}}$ without loss of
 generality.
If $|w| < N$ it is clear that
 \[
w \in  \wmix{w_1, \ldots, w_k} \cap A^{<N} \neq \wmix{w'_1, \ldots, w'_{k'}} \cap
 A^{<N} \not\ni w
 \]
 and the ``if''-part holds.
 Thus we consider the case $|w| \geq N$.
 Let $w = vu$ such that $|v| = N$ and
 $M = \mtr(\Walk(v, u))$.
 We now prove that $\wmix{w'_1, \ldots, w'_{k'}}$ does not contain any word
 $w' = v'u' \, (|v'| = N)$ that has the same multi-trace with $w$ (\ie,
 $\mtr(\Walk(v', u')) = M = \mtr(\Walk(v, u))$; $v' = v$ holds in this case).
 By Proposition~\ref{prop:dg} and Proposition~\ref{prop:mtrace}, any subword occurrences
 in a word is completely determined by its multi-trace.
 Thus if there is a word $w' = vu'$ in $\wmix{w'_1, \ldots,
 w'_{k'}}$ such that $\mtr(\Walk(v, u')) = M$, then
\begin{align*}
\occv{w'}{w'_1, \ldots, w'_{k'}} = &\, \occv{v}{w'_1, \ldots, w'_{k'}} + \!\!\!\!\!\!
 \sum_{\path \in \paths(\dg)} \!\!\!\!\!
 M(\path) \cdot \occv{\path}{w'_1, \ldots, w'_{k'}}
 + \!\!\!\!\!\! \sum_{\cycle \in \cycles(\dg)} \!\!\!\!\!
 M(\cycle) \cdot \occv{\cycle}{w'_1, \ldots, w'_{k'}}
 \\
 = &\, \occv{w}{w'_1, \ldots, w'_{k'}}
\end{align*}
 from which we obtain $w \in \wmix{w'_1, \ldots, w'_{k'}}$; this contradicts with the assumption.
Therefore we can conclude that
\[
 M \in \ltr(\wmix{w_1, \ldots, w_k}) \neq
 \ltr(\wmix{w'_1, \ldots, w'_{k'}}) \not\ni M. \quad\qed
\]
\end{proof}

\subsection{Decidability}
By using Theorem~\ref{thm:equiv}, we can obtain an algorithm for deciding the equivalence of two WMIX languages.
This algorithm also uses the decidability of Presburger arithmetic, as
like the previous algorithm for the infiniteness,
but in contrast to the case of inifiniteness,
\emph{it is reduced into $\mathrm{\Pi}_1$-formula of Presburger arithmetic}.

\begin{corollary}\label{cor:equiv}
For any word $w_1, \ldots, w_k, w'_1, \ldots, w'_{k'} \in A^*$, it is
 decidable whether $\wmix{w_1, \ldots, w_k} = \wmix{w'_1, \ldots, w'_{k'}}$ or not.
\end{corollary}
\begin{proof}
Let $N = \max(|w_1|, \ldots, |w_k|, |w'_1|, \ldots, |w'_{k'}|)$.
We can effectively check $\wmix{w_1, \ldots, w_k} \cap A^{<N} = \wmix{w'_1,
 \ldots, w'_{k'}} \cap A^{<N}$ holds or not, since
 $A^{<N}$ is finite.
If $\wmix{w_1, \ldots, w_k} \cap A^{<N} \neq \wmix{w'_1, \ldots, w'_{k'}} \cap
 A^{<N}$ then two languages are not equivalent.
Otherwise, enumerate all possible traces in $\dg$, then for each trace
 $T$, and check whether every multi-trace $M$ with $T = \{ \walk \in \paths(\dg) \cup \cycles(\dg) \mid M(\walk) \neq 0\}$ satisfies
\begin{align}
 M \in \ltr(\wmix{w_1, \ldots, w_k}) \text{ if and only if }
 M \in \ltr(\wmix{w'_1, \ldots, w'_{k'}}) \tag{$\blacklozenge$}\label{lozen}
\end{align}
 or not.
If there is some multi-trace that does not satisfy Condition~\eqref{lozen} then $\wmix{w_1, \ldots, w_k} \neq \wmix{w'_1, \ldots,
 w'_{k'}}$, otherwise $\wmix{w_1, \ldots, w_k} = \wmix{w'_1, \ldots, w'_{k'}}$
 holds.
Since every multi-trace can be represented by a corresponding trace and
 its multiplicity (positive coefficients),
 for a trace $T = \{\path \} \cup \{\cycle_1, \ldots, \cycle_m\}$, the
 statement ``every multi-trace $M$ with $T = \{ \walk \in \paths(\dg)
 \cup \cycles(\dg) \mid M(\walk) \neq 0\}$ satisfies
 Condition~\eqref{lozen}'' can be represented by the following
 $\mathrm{\Pi}_1$-formula of Presburger arithmetic $\psi_T$:
\begin{align*}
 \psi_T \defeq \, & \forall x_1, \ldots, x_m \,
 (x_1 > 0 \land \cdots \land x_m > 0)\\
& \Rightarrow \left( \Bigl(\theta_{T}^{w_1} =
 \cdots = \theta_{T}^{w_k}
 \Bigr) \Leftrightarrow
 \Bigl( \theta_{T}^{w'_1} =
 \cdots = \theta_{T}^{w'_{k'}}
 \Bigr) \right)
\end{align*}
 where $\theta_T^{w}$ is a subexpression defined by
 \[
 \theta_T^{w} \defeq \occv{\from{\path}}{w} + \occv{\path}{w} +
 \sum_{i = 1}^{m} \underbrace{x_i + \cdots + x_i}_{\occv{\cycle_i}{w}
 \text{ times}}.  \qquad\qed
 \]
\end{proof}
\section{Open Problem}\label{sec:future}
We would like to introduce the following open
problem which asks the existence of a non-trivial finite WMIX language.
\begin{problem}[\cite{problem}]
Are there $w_1, \ldots w_k \in A^*$ such  that
 $\wmix{w_1, \ldots, w_k}$ is finite but
 $|w|_{w_1} = \cdots = |w|_{w_k} \geq 1$ for some $w \in \wmix{w_1, \ldots, w_k}$?
\end{problem}
Note that all examples of finite WMIX languages in this note are not of
this type.
The complexity issue is also interesting.
\begin{problem}
What is the complexity of the infiniteness 
 problem (resp. the equivalence problem) for WMIX languages?
\end{problem}

\vspace{1cm}
\noindent{\bf Acknowledgment.}
The author would like to thank Thomas Finn Lidbetter (University of
Waterloo) for telling me this topic.
The author also thank to an anonymous reviewer for letting me know
some known results on unambiguous CA~\cite{UnCA} and pointing out
that the decidability results presented in this note are also from those
results.

\bibliographystyle{splncs.bst}
\bibliography{ref}

\begin{thebibliography}{10}

\bibitem{joshi}
Joshi, A., Vijay-Shanker, K., Weir, D.
\newblock The Convergence of Mildly Context-sensitive Grammar Formalisms.
  Foundational Issues in Natural Language Processing (1991)  31--82 cited By 1.

\bibitem{marsh}
Marsh, W.:
\newblock Some conjectures on indexed languages.
\newblock Abstract appears in {\it Journal of Symbolic Logic} \textbf{51}(3)
  (1985)  849 Paper presented to the Association for Symbolic Logic Meeting,
  Stanford University, July 15--19, 1985.

\bibitem{rcg}
Boullier, P.:
\newblock Chinese numbers, mix, scrambling, and range concatenation grammars.
\newblock In: Proceedings of the Ninth Conference on European Chapter of the
  Association for Computational Linguistics. EACL '99, Stroudsburg, PA, USA,
  Association for Computational Linguistics (1999)  53--60

\bibitem{kanazawa}
Kanazawa, M., Salvati, S.:
\newblock Mix is not a tree-adjoining language.
\newblock In: Proceedings of the 50th Annual Meeting of the Association for
  Computational Linguistics: Long Papers - Volume 1. ACL '12, Stroudsburg, PA,
  USA, Association for Computational Linguistics (2012)  666--674

\bibitem{salvati}
Salvati, S.:
\newblock Mix is a 2-mcfl and the word problem in z2 is captured by the io and
  the oi hierarchies.
\newblock Journal of Computer and System Sciences \textbf{81}(7) (2015)  1252
  -- 1277

\bibitem{sorokin}
Sorokin, A.:
\newblock Ogden property for linear displacement context-free grammars.
\newblock In Artemov, S., Nerode, A., eds.: Logical Foundations of Computer
  Science, Cham, Springer International Publishing (2016)  376--391

\bibitem{parikh}
Parikh, R.J.:
\newblock On context-free languages.
\newblock J. ACM \textbf{13}(4) (October 1966)  570--581

\bibitem{Finn}
Colbourn, C.J., Dougherty, R.E., Lidbetter, T.F., Shallit, J.:
\newblock Counting subwords and regular languages.
\newblock In: Developments in Language Theory - 22nd International Conference,
  {DLT} 2018, Tokyo, Japan, September 10-14, 2018, Proceedings. (2018)
  231--242

\bibitem{parikhmat}
Mateescu, A., Salomaa, A., Salomaa, K., Yu, S.:
\newblock A sharpening of the parikh mapping.
\newblock Theoretical Informatics and Applications \textbf{35}(6) (2001)
  551--564

\bibitem{UnCA}
Cadilhac, M., Finkel, A., McKenzie, P.:
\newblock Unambiguous constrained automata.
\newblock In Yen, H.C., Ibarra, O.H., eds.: Developments in Language Theory,
  Berlin, Heidelberg, Springer Berlin Heidelberg (2012)  239--250

\bibitem{deBruijn}
de~Bruijn, N.G.:
\newblock {A Combinatorial Problem}.
\newblock Koninklijke Nederlandsche Akademie Van Wetenschappen \textbf{49}(6)
  (June 1946)  758--764

\bibitem{deLuca}
de~Luca, A., Varricchio, S.:
\newblock Finiteness and Regularity in Semigroups and Formal Languages. 1st
  edn.
\newblock Springer Publishing Company, Incorporated (2011)

\bibitem{sipser}
Sipser, M.:
\newblock Introduction to the theory of computation. 3 edn.
\newblock Cengage Learning (2012)

\bibitem{problem}
Lidbetter, T.F.:
\newblock {\it personal communication} (2018)

\end{thebibliography}

\end{document}